\documentclass{article}
\usepackage[utf8]{inputenc}
\usepackage{amsmath}
\usepackage{amsthm}
\usepackage{amsfonts}
\usepackage{amssymb}
\usepackage{graphicx}
\usepackage{tikz}
\usetikzlibrary{shapes.geometric, arrows}

\tikzstyle{list} = [rectangle, minimum width=2cm, minimum height=0.8cm, text centered, draw=black]
\tikzstyle{node} = [rectangle, text centered, minimum width = .7cm, draw=black]
\tikzstyle{arrow} = [thick, ->, >=stealth]
\tikzstyle{darrow} = [thick, <->]

\newtheorem{theorem}{Theorem}
\newtheorem{corollary}{Corollary}[theorem]

\theoremstyle{definition}
\newtheorem{definition}{Definition}

\theoremstyle{remark}

\title{List Heaps}
\author{Andrew Frohmader\thanks{andrewfrohmader@gmail.com}}
\date{February 2018}

\begin{document}

\maketitle

\begin{abstract}
    This paper presents a simple extension of the binary heap, the List Heap. We use List Heaps to demonstrate the idea of \textit{adaptive heaps}: heaps whose performance is a function of both the size of the problem instance and the disorder of the problem instance. We focus on the presortedness of the input sequence as a measure of disorder for the problem instance. A number of practical applications that rely on heaps deal with input that is not random. Even random input contains presorted subsequences. Devising heaps that exploit this structure may provide a means for improving practical performance. We present some basic empirical tests to support this claim. Additionally, adaptive heaps may provide an interesting direction for theoretical investigation.
\end{abstract}

\section{Introduction} \label{introduction}
    A heap is a data structure which holds a finite set of items. Each item is associated with a key drawn from a totally ordered set. Heaps support the following operations:

\begin{table}[h]
    \begin{tabular}{p{0.3\linewidth} p{0.6\linewidth}}
        \textit{make heap (h):} & Create and return a new, empty heap $h$\\
        \textit{insert (h, x, k):} & Insert item $x$ with key $k$ into heap $h$ and return a reference to where $x$ is stored in $h$\\
        \textit{find min (h):} & Return a reference to where the item with the minimum key is stored in heap $h$\\
        \textit{delete min (h):} & Delete the item with the minimum key from heap $h$ and return it\\
        \textit{decrease key (h, x, k):} & Decrease the key of item $x$ in heap $h$ to $k$\\
        \textit{delete (h, x):} & Delete item $x$ from heap $h$\\
        \textit{meld ($h_{1}$, $h_{2}$):} & Return a heap formed by taking the union of heaps $h_{1}$ and $h_{2}$\\
    \end{tabular}
\end{table}

    The binary heap was introduced by Williams in 1964 \cite{binary:heap:1964}. Its simplicity and speed have made it and its generalization, the d-array heap, a popular choice in practice. It supports \textit{insert}, \textit{delete min}, and \textit{decrease key} in $O(\log{n})$ time. It can be used to sort $n$ items in $O(n\log{n})$, which matches the worst-case lower bound for a comparison sort. Vuillemin's introduction of the binomial queue in 1978 \cite{binomial:queue:1978}, added \textit{meld} to the list of operations supported in $O(\log{n})$.
    
    In 1984, Fibonacci heaps \cite{fibonacci:heap:1987}, an extension of the binomial queue, achieved $O(1)$ amortized time for \textit{insert}, \textit{decrease key}, and \textit{meld}. The \textit{decrease key} result was particularly important in that it improved the worst-case bounds for a number of well-known graph algorithms. More recently, a few structures have achieved worst-case $O(1)$ time for \textit{decrease key} and \textit{meld}, see \cite{worst:case:efficient:1996} or \cite{strict:fibonacci:heap:2012}. While this work has produced interesting and important theoretical results, it has failed to yield a structure that consistently outperforms the original binary heap and its variants in practice \cite{empirical:study}.

    In this paper, we return to the binary heap and develop a simple extension, the List Heap. This straightforward extension can be given \textit{adaptive} operations: operations whose performance depends not only on the problem size, but also on the level of presortedness (disorder) in the problem instance.  A bit of work has gone into developing the theory of adaptive sorting algorithms, see \cite{adaptive:sorting:survey:1992}, but to our knowledge, this work has not migrated into the related work on heap data structures, \cite{priority:queue:survey}. We believe that \textit{adaptive heaps} may provide an interesting angle for theoretical investigation. Additionally, they may provide a means of improving the empirical performance of current heap variants. The List Heap is a first step in this direction. 
    
    List Heaps support \textit{decrease key}, \textit{insert}, and \textit{delete min} in $O(\log{k})$, where $k$ is the number of lists in the List Heap. As we will show, the number of lists in a List Heap is a function of both the size of the problem instance and the disorder of the problem instance. We returned to the binary heap because of its simplicity and ubiquity, but this was not without costs. List Heaps lose the $O(1)$ \textit{insert}, \textit{decrease key}, and \textit{meld} of more sophisticated structures.
    
     \subsection{Preliminaries} \label{preliminaries}
    
    Here we present notational conventions and definitions used through the remainder of this paper. Let $X=\langle x_{1},...,x_{n}\rangle$ be a sequence of $n$ distinct elements $x_{i}$ from some totally ordered set. If $x_{1}<x_{2}<...<x_{n}$, $X$ is \textit{monotonically increasing} or just \textit{increasing}. If $x_{1}>x_{2}>...>x_{n}$, $X$ is \textit{monotonically decreasing} or just \textit{decreasing}. A sequence is $monotonic$ if it is either increasing or decreasing. The $head$ of a sequence $X$ is $x_{1}$, the $tail$ is $x_{n}$. If $A$ is a set, then $||A||$ is its cardinality. If $X$ is a sequence, then $|X|$ is its length. For two sequences $X=\langle x_{1},...,x_{n}\rangle$ and $Y=\langle y_{1},...,y_{m}\rangle$, their $concatenation$ $XY$ is the sequence $\langle x_{1},...,x_{n},y_{1},...,y_{m}\rangle$. If the sequence $X$ contains no elements, we write $X=\emptyset$.
    
    A sequence obtained by deleting zero or more elements from $X$ is called a $subsequence$ of $X$. A subsequence $Y=\langle x_{i},...,x_{j}\rangle$ of $X$ is $consecutive$ if the indices $i,...,j$ are consecutive integers.
    
    Let $Y = \langle x_{i},...,x_{j} \rangle$ and $Z = \langle x_{k},...,x_{l} \rangle$ be subsequences of $X$. The \textit{intersection} of $Y$ and $Z$, $Y \cap Z$, is the subsequence of $X$ obtained by deleting from $X$ all $x_{h}$ not in both $Y$ and $Z$ for $1 \leq h \leq n$. Similarly, the $union$ of $Y$ and $Z$, $Y \cup Z$, is the subsequence of $X$ obtained by deleting from $X$ all $x_{h}$ not in either $Y$ or $Z$ for $1 \leq h \leq n$. $Y$ and $Z$ are \textit{disjoint} if $Y \cap Z = \emptyset$. Let $P = \{X_{1},...,X_{k}\}$ be a set of disjoint subsequences of $X$, if the union of all subsequences in $P$ equals $X$, then $P$ is a \textit{partition} of $X$.

    \subsection{Adaptive Sorting} \label{adaptive_sorting}
    
    This section gives a very brief review of adaptive sorting. Heaps solve a generalized sorting problem, so adaptive sorting provides some intuition for why adaptive heaps might be useful. For a more detailed survey of adaptive sorting, see \cite{adaptive:sorting:survey:1992} or \cite{framework:1995}.

    Consider the sorting problem: take as input some arbitrary sequence $X=\langle x_{1},...,x_{n}\rangle$ of elements from a totally ordered set and return a permutation of the sequence that is in increasing sorted order. Comparison based sorting has a well-know worst-case lower bound of $\Omega(n\log{n})$ \cite{CLRS}. However, it is clear that this lower bound must not always hold. What if our input sequence is already sorted? What if only one element is out of place? What if it is the concatenation of two sorted subsequences? The lower bound can be refined if we account for the disorder in the input sequence.
    
    The main achievements of the adaptive sorting literature are: proposing a variety of measures of disorder, proving new lower bounds with respect to these measures, developing sorting algorithms whose performance matches these new lower bounds, and developing a partial order on the set of measures.
    
    We stop here and again direct the reader to \cite{adaptive:sorting:survey:1992} or \cite{framework:1995} for more information.
    
    \subsection{Outline of Paper}
    
    The remainder of the paper is organized as follows. Section \ref{Sec:Adaptive Heaps} discusses why adaptive heaps might be worth developing. Section \ref{list_heaps} presents List Heaps - their structure and operations. Section \ref{Empirical Results} presents the results of a series of brief empirical tests suggesting List Heaps may have promise in practice. Section \ref{Conclusion} summarizes results obtained.
    
\section{Adaptive Heaps} \label{Sec:Adaptive Heaps}
    
    This section presents a few reasons for developing adaptive heaps.
    
    We use the term \textit{adaptive heap} loosely throughout this paper to refer to a heap whose performance is some function of the level of presortedness (disorder) of the input sequence. There are clearly complications we are glossing over, the largest of which is how to deal with \textit{decrease key}. Further, heap problem instances can have disorder of different types, not just related to the presortedness of the input sequence. A starting point to formalize these notions is the adaptive sorting literature. It is fairly easy to extend the results of adaptive sorting to arbitrary sequences of \textit{insert} and \textit{delete min}; \textit{decrease key} may prove more challenging. 
    
    Why build a heap whose performance adapts to the presortedness of the input sequence? Just as the lower bound for sorting can be restated as a function of the presortedness of the input sequence, the lower bound for an arbitrary heap problem can be restated as a function of the disorder of the problem instance. As \textit{delete min} is the only operation for which $O(1)$ performance is not available, the disorder of the problem directly impacts the bound on \textit{delete min}. Thus, we can restate the bound on \textit{delete min} as $O(\log{k})$ where $k$ is some measure of disorder less than or equal to the number of elements in the heap. Then, intuitively, the bound on \textit{delete min} ranges from $O(1)$ for sorted input to $O(\log{n})$ for random input.
    
    The bounds on algorithms which rely on heaps, such as Dijkstra's shortest path algorithm \cite{Dijkstra}, can similarly be adjusted to reflect disorder. For example, we can restate the Dijkstra bound as $O(n\log{k} + m)$ where $n$ is the number of vertices, $m$ is the number of edges, and $k$ is a measure of disorder with $k\leq n$. 
    
    The discussion above suggests that heap-based applications that have input sequences with some level of presortedness could benefit from adaptivity. What is less obvious is that even heap-based applications that have random input sequence can benefit from adaptivity. Here, the benefit is not in asymptotic performance, but in constant factors. Let $X$ be a random sequence of $n$ elements. As $n \rightarrow \infty$, the minimum number of increasing subsequences into which $X$ can be partitioned is approximately $2\sqrt{n}$ subsequences \cite{Aldous99longestincreasing}. As we will show in Section \ref{Enc Adaptive}, we can create an extension of the binary heap that is adaptive to the minimum number of increasing subsequences. Thus, where a binary heap constructed from $X$ performs roughly $2\log{n}$ comparisons to \textit{delete min} in the worst-case, an adaptive heap constructed from $X$ can perform \textit{delete min} in worst-case $\log{n}$ comparisons - cutting the constant factor in half.
    
\section{List Heaps} \label{list_heaps}

    This section introduces the List Heap. The List Heap structure closely mirrors that of the binary heap. The hope is that this similarity makes the changes required to add adaptivity to heaps clear. Additionally, we believe adaptive variants of the binary heap have the greatest potential to be immediately useful in practice. However, binary heaps are flawed as a choice in that \textit{decrease key} and \textit{insert} take $O(\log{n})$ time while $O(1)$ implementations of these operations are clearly possible. Thus the List Heap cannot be optimal on all problem instances. We tolerate this flaw and show that we can construct a heap whose performance is a function of a number of measures of disorder. In particular, we will focus on developing heaps that are adaptive to $runs$, $SUS$, and $Enc$ (defined later). These measures of disorder partition the input into monotonic subsequences which are particularly simple for a heap to use.
    
    This paper focuses only on \textit{insert}, \textit{delete min}, and \textit{decrease key} and assumes all keys are unique. We first consider the structure of the heap. Next, we outline operations that are adaptive with respect to $runs$. Finally, we present operations that are adaptive with respect to $SUS$ and $Enc$.
    
    \subsection{Structure}
    
    A \textit{List Heap} $h$ is an array of circular doubly linked lists of nodes, $h = \langle l_1, ... , l_k \rangle$. Each node has a unique key associated with it. Throughout this paper we will refer to the node $x$ and its key interchangeably. We use the notation $x_{(i,j)}$ to refer to the $jth$ node in list $l_{i}$. List Heaps must maintain two invariants. 
    
    \begin{enumerate}
        \item  For any list $l_i$ in heap $h$, the nodes of $l_i$ are in increasing sorted order by key value, $l_i = \langle x_{(i,1)}, ... , x_{(i,j)} \rangle$, with $x_{(i,1)} < x_{(i,2)} < ... < x_{(i,j)}$.
        \item The lists are arranged in \textit{heap order} in the array, that is $l_{\lfloor i / 2 \rfloor} < l_i$ where we define $l_i < l_j$ if $x_{(i,1)} < x_{(j,1)}$.
    \end{enumerate}
    
    With the invariants above, we can view a List Heap as a standard binary heap of elements: each binary heap element is a list and the key for each list is the key of its head node. With this structure, the minimum node is $x_{(1,1)}$ - the head node of list $l_1$. We will refer to $l_1$ as the \textit{root list} of the heap.
    
    \begin{figure}[h]
        \centering
        \begin{tikzpicture}[node distance=.8cm]
    
        \node (l0) [list] {};
        \node (l1) [list, below of=l0] {$l_1$};
        \node (l2) [list, below of=l1] {$l_2$};
        \node (l3) [list, below of=l2] {$l_3$};
        \node (l4) [list, below of=l3] {$l_4$};
        \node (l5) [list, below of=l4] {$l_5$};
        \node (l6) [list, below of=l5] {$l_6$};
        \node (l7) [list, below of=l6] {$l_7$};
        
        \node[text width=1.3cm, left of=l0] {0};
        \node[text width=1.3cm, left of=l1] {1};
        \node[text width=1.3cm, left of=l2] {2};
        \node[text width=1.3cm, left of=l3] {3};
        \node[text width=1.3cm, left of=l4] {4};
        \node[text width=1.3cm, left of=l5] {5};
        \node[text width=1.3cm, left of=l6] {6};
        \node[text width=1.3cm, left of=l7] {7};
        
        \node (n1) [node, right of=l1, xshift=1.2cm] {1};
        \node (n3) [node, right of=n1, xshift=.5cm] {3};
        
        \node (n4) [node, right of=l2, xshift=1.2cm] {4};
        \node (n14) [node, right of=n4, xshift=.5cm] {14};
        \node (n15) [node, right of=n14, xshift=.5cm] {15};
        
        \node (n2) [node, right of=l3, xshift=1.2cm] {2};
        \node (n8) [node, right of=n2, xshift=.5cm] {8};
        \node (n9) [node, right of=n8, xshift=.5cm] {9};
        
        \node (n5) [node, right of=l4, xshift=1.2cm] {5};
        \node (n13) [node, right of=n5, xshift=.5cm] {13};
        
        \node (n6) [node, right of=l5, xshift=1.2cm] {6};
        \node (n10) [node, right of=n6, xshift=.5cm] {10};
        \node (n12) [node, right of=n10, xshift=.5cm] {12};
        
        \node (n11) [node, right of=l6, xshift=1.2cm] {11};

        \node (n7) [node, right of=l7, xshift=1.2cm] {7};
        \node (n16) [node, right of=n7, xshift=.5cm] {16};

        \draw [arrow] (l1) -- (n1);
        \draw[darrow] (n1) -- (n3);
        \draw[darrow] (n1) to [out=155,in=25] (n3);
        
        \draw [arrow] (l2) -- (n4);
        \draw [darrow] (n4) -- (n14);
        \draw [darrow] (n14) -- (n15);
        \draw[darrow] (n4) to [out=160,in=20] (n15);
        
        \draw [arrow] (l3) -- (n2);
        \draw [darrow] (n2) -- (n8);
            \draw [darrow] (n8) -- (n9);
        \draw[darrow] (n2) to [out=160,in=20] (n9);
    
        \draw [arrow] (l4) -- (n5);
        \draw [darrow] (n5) -- (n13);
        \draw[darrow] (n5) to [out=155,in=25] (n13);
    
        \draw [arrow] (l5) -- (n6);
        \draw [darrow] (n6) -- (n10);
        \draw [darrow] (n10) -- (n12);
        \draw[darrow] (n6) to [out=160,in=20] (n12);
    
        \draw [arrow] (l6) -- (n11);
        \draw[darrow] (n11) to [out=140,in=40] (n11);
    
        \draw [arrow] (l7) -- (n7);
        \draw [darrow] (n7) -- (n16);
        \draw[darrow] (n7) to [out=155,in=25] (n16);
    
        \end{tikzpicture}  
        \caption{The List Heap Structure}
        \label{fig:list_heap}
    \end{figure}
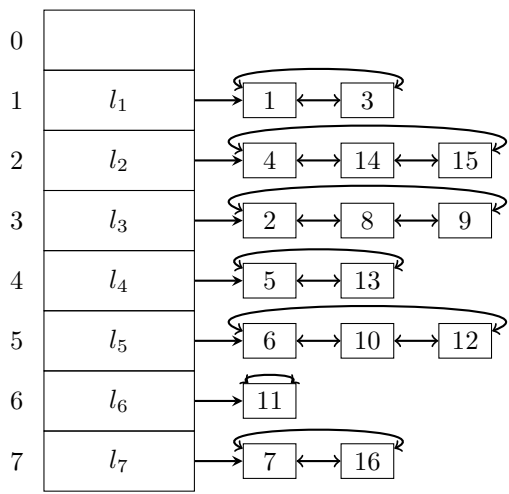
    
    The sorted linked lists are the key feature of the List Heap. They enable us to capture some of the existing order in the input sequence. As a result, List Heaps can be given operations that are adaptive with respect to some measures of disorder.
    
    \subsection{A Runs Adaptive List Heap}
    
    In this section, we present operations for a List Heap that is adaptive with respect to $runs$, abbreviated RA List Heap. We show that the performance of each heap operation is a function of the number of lists in the heap and that the \textit{insert} operation partitions the sequence of inserted elements into $L$ lists where $L$ is less than or equal to the number of $runs$ in the sequence. 
    
    \begin{definition}
        A $run$ is a consecutive decreasing subsequence of maximum length.
    \end{definition}
    
    A subsequence $X_i = \langle x_j,...,x_k \rangle$ of $X$ is a run if and only if $j, ..., k$ are consecutive integers and  $x_{j-1} < x_j > ... > x_k < x_{k+1}$. The total number of $runs$ in a sequence gives a measure of disorder. For example, consider the following random sequence, 
    
    \begin{equation} \label{eq:random sequence}
        \langle 3,15,14,4,9,13,5,12,10,6,1,11,8,16,2,7 \rangle.
    \end{equation}
    
    This sequence of sixteen elements can be partitioned into eight runs,
    
    \begin{equation}
        \langle 3 \rangle, \langle 15,14,4 \rangle, \langle 9 \rangle, \langle 13,5 \rangle, \langle 12, 10, 6, 1 \rangle, \langle 11, 8 \rangle, \langle 16, 2 \rangle, \langle 7 \rangle.
    \end{equation}
    
    A decreasing sequence will consist of just one run, an increasing sequence of $n$ elements will contain $n$ runs, each with only one element. We now outline operations for the RA List Heap.

    \subsubsection{Insert}
        To \textit{insert} a new item $x$ into heap $h = \langle l_1, ... , l_k \rangle$, do the following. If $x$ is less than the head node of $l_k$, (recall that we denote this $x < l_k$), set $i = k$. Otherwise, set $i = k+1$. Then, while $x < l_{\lfloor i/2 \rfloor}$, set $i = \lfloor i/2 \rfloor$ and repeat. Once the loop terminates, append $x$ to the front of list $l_i$. Note that $i$ might equal $k+1$ in which case a new empty list $l_{k+1}$ will need to be created before $x$ is appended. 
        
        Figure \ref{fig:list_heap} shows the RA List Heap generated by inserting the random sequence (\ref{eq:random sequence}) into an empty heap. 
    \subsubsection{Delete Min}
        To perform \textit{delete min} on heap $h$, remove the head node $x_{(1,1)}$ from the root list $l_1$ and set $x_{(1,1)}$ to the return value of the function. If $l_1$ is now empty, replace it with the last list in the heap $l_k$. If $l_1$ is not empty, there is no need to swap it with $l_k$. Either way, at this point $l_1$ might be out of heap order. Restore heap order by calling \textit{heapify down}. This is identical to the heapify operation of the binary heap except that it manipulates entire lists instead of individual nodes. To \textit{heapify down} on a list $l_i$, compare $l_i$ with its two children, $l_{2i}$ and $l_{2i+1}$. If $l_i$ is less than both children, do nothing. Otherwise, swap $l_i$ with its smallest child and repeat. 
    \subsubsection{Decrease Key}
    
        To perform \textit{decrease key} on a node $x_{(i,j)}$, set $x_{(i,j)}$'s key to the new value. Now there are two cases.
        
        \begin{enumerate}
            \item If $x_{(i,j)}$ is a head node, that is if $x_{(i,j)} = x_{(i,1)}$, decreasing $x_{(i,1)}$'s key maintains the sorted order of the list $l_i$ but may destroy the heap order of the array. Call \textit{heapify up} to restore the heap order. To heapify up a list $l_i$, compare $l_i$ to its parent $l_{\lfloor i/2 \rfloor}$. If $l_i < l_{\lfloor i/2 \rfloor}$ swap the lists. Recurse until heap order is restored.
            \item  If $x_{(i,j)}$ is not a head node, $l_i$ may no longer be sorted. If $x_{(i,j)}$ is less than its left sibling $x_{(i,j-1)}$, then the list $l_i$ is out of sorted order. Remove $x_{(i,j)}$ from $l_i$ and reinsert $x_{(i,j)}$ into the List Heap using the \textit{insert} routine.
        \end{enumerate}
    
    \subsubsection{Analysis}
        
        It is clear from their descriptions that each of these functions run in $O(\log{k})$ time, where $k$ is the number of lists in the RA List Heap. 
        
        \begin{theorem} \label{thm:insert_runs}
            A sequence of consecutive insert operations partitions the inserted elements into $L$ lists where $L$ is less than or equal to the minimum number of runs in the input sequence.
        \end{theorem}
        
        \begin{proof}
            Say the consecutive sequence of elements inserted into heap $h = \langle l_1, ... l_k \rangle$ is $X = \langle x_1,..., x_i, x_{i+1}, ... \rangle$. We must show that if $x_i > x_{i+1}$, then $x_{i+1}$ is appended to an existing list. There are a couple of cases to consider.
            
            \begin{enumerate}
                \item If $x_i < l_k$, $x_i$ is appended to $l_k$ or one of $l_k$'s ancestors. Now, $x_{i+1} < x_i$ so $x_{i+1} < l_k$ and $x_{i+1}$ is also appended to $l_k$ or one of $l_k$'s ancestors.
                
                \item If $x_i > l_k$, $x_i$ is appended to $l_{k+1}$ or one of its ancestors.
                
                If $x_i$ is appended to $l_{k+1}$, the new list $l_{k+1}$ must be created. Now, $x_{i+1}$ is compared first to $l_{k+1}$. Since $x_{i+1} < x_i$, $x_{i+1} < l_{k+1}$, thus $x_{i+1}$ is appended to $l_{k+1}$ or one of $l_{k+1}$'s ancestors. 
                
                If instead $x_i$ was appended to an ancestor of $l_{k+1}$, $l_{k+1}$ was not created. Then $x_{i+1}$ is compared first to $l_k$. If $x_{i+1} > l_k$, $x_{i+1}$ is compared to the ancestors of $l_{k+1}$. Since $x_i$ was appended to an ancestor of $l_{k+1}$, $x_{i+1}$ must be less than one of the ancestors, so $x_{i+1}$ is appended to an existing list.
            \end{enumerate}

        \end{proof}
    
    \subsection{An Enc Adaptive List Heap} \label{Enc Adaptive}
    
    In this section we present operations for a List Heap that is adaptive with respect to $runs$, $SUS$, and $Enc$. We will abbreviate this EA List Heap. $SUS$ is due to Levcopoulos and Petersson \cite{SUS:ENC:SMS:1994}, $Enc$ was proposed by Skiena \cite{encroaching}.
    
    \begin{definition}
        $SUS$ \textit{(Shuffled UpSubsequences)} is the minimum number of increasing subsequences into which a sequence $X$ can be partitioned.
    \end{definition}
    
    This differs from runs in that the subsequences are increasing and they are not required to be consecutive. For example, the sequence (\ref{eq:random sequence}) can be partitioned into the following seven increasing subsequences, 
    
    \begin{equation}
        \langle 3, 15, 16 \rangle, \langle 14 \rangle, \langle 4,9,13 \rangle, \langle 5, 12 \rangle, \langle 10, 11 \rangle, \langle 6, 8 \rangle, \langle 1,2,7 \rangle.
    \end{equation}
    
    This is, in fact, an optimal partition, so $SUS$ for the sequence (\ref{eq:random sequence}) is seven.

    \begin{definition}
        An \textit{encroaching set} is an ordered set of increasing sequences $\langle E_{1},...,E_{k}\rangle$ such that the $head(E_{i}) < head(E_{i+1})$ and $tail(E_{i}) > tail(E_{i+1})$ for $1\leq i < k$.
    \end{definition}
    
     Thus, the increasing sequences nest or $encroach$ upon one another. Skiena describes an encroaching set by an algorithm, \textit{melsort}, that builds them. Given an input sequence $Y = \langle y_1, ... , y_n \rangle$ that is a permutation of an ordered set, build the encroaching set $E = \langle E_{1},...,E_{k}\rangle$ as follows. If the item $y_i$ fits on either end of one of the increasing sequences $E_j$, put it there. Otherwise, form a new sequence $E_{k+1}$. Place each item on the oldest $E_j$ upon which it fits. We note that this strategy bears a lot of similarities to Patience Sorting, see \cite{Aldous99longestincreasing}. The encroaching set for our example sequence (\ref{eq:random sequence}) is, 
     
     \begin{equation}
        \langle 1, 3, 15, 16 \rangle, \langle 2, 4, 14 \rangle, \langle 5,9,13 \rangle, \langle 6, 10, 12 \rangle, \langle 7, 8, 11 \rangle.
     \end{equation}
     
     Thus, $Enc$ for the sequence (\ref{eq:random sequence}) is five. For adaptive sorting, an $Enc$ optimal algorithm is also $SUS$ optimal and $runs$ optimal \cite{framework:1995}. Therefore, we focus on developing a heap which is adaptive to $Enc$, thereby also achieving adaptivity with respect to $SUS$ and $runs$.
     
     Again, we emphasize that this paper treats the notation of adaptivity informally. In order to achieve some level of adaptivity with respect to $Enc$, we only need to alter our \textit{insert} function. The \textit{decrease key} and \textit{delete min} functions can remain the same as above. 
    
    \subsubsection{Insert} \label{enc:insert}
    
     The \textit{insert} function attempts to build an encroaching set. 
     
     We adapt Skiena's strategy to the dynamic context of heaps where there is not a fixed input sequence. Let $y$ be a new item to insert and $h = \langle l_1, ... , l_k \rangle$ be the heap. In this context, $l_j$ is roughly equal to $E_j$ of the encroaching list set. There are three cases.
     \begin{enumerate}
         \item Try to insert $y$ at the front of a list in $h$. Use binary search on the lists to find the list that has $y < l_i$ and minimizes $|y - l_i|$. Note, the lists are maintained in heap order, not sorted order. As a result, this binary search may only be an approximate heuristic. Moreover, the binary search may have failed to consider $l_i$'s direct parent. Appending $y$ to the front of $l_i$ without checking $l_i$'s parent could destroy heap order. Thus, we must also compare $y$ to $l_i$'s ancestors. If $y < l_{\lfloor i/2 \rfloor}$, set $i$ = $\lfloor i/2 \rfloor$ and repeat. If we find an $l_i$, we are done. Otherwise, we move on to Case 2.
         \item Try to insert $y$ at the tail of a list in $h$. Use binary search on the tail nodes to find the list $l_i$ that has $tail(l_i) < y$ and minimizes $|y- tail(l_i)|$. If no $l_i$ is found, go to Case 3.
         \item Create a new list $l_{k+1}$ with the new node as its only element and insert the list after $l_k$ in the array. 
     \end{enumerate}
    
    \subsubsection{Analysis}
    
    It is easy to see that the \textit{insert} routine described above runs in $O(\log{k})$, where $k$ is the number of lists in the EA List Heap. 
    
    \begin{theorem} \label{Thm:Enc Adaptive Encroaching}
    A sequence of $n$ consecutive inserts to an empty EA List Heap h partitions the input into an optimal encroaching set.
    \end{theorem}
    
    \begin{proof}
        Simply observe that the \textit{insert} function reduces to the list creation phase of \textit{melsort} if the heap starts empty.
        The only algorithmic difference between the \textit{insert} routine above and Skiena's \textit{melsort} is the addition of the heap order check. 
    \end{proof}
    
    \begin{corollary}
        Given an empty EA List Heap $h$, perform $n$ consecutive insert operations of elements with uniformly random keys. Then, as $n$ grows large, the expected worst-case number of comparisons to delete the minimum item from $h$ is less than or equal to half the the number of comparisons required by a binary heap.
    \end{corollary}
    
    \begin{proof}
        By Theorem \ref{Thm:Enc Adaptive Encroaching}, $h = \langle l_1, ... , l_k \rangle$ where the lists $l_i$ form an optimal encroaching set. As mentioned above, an $Enc$ optimal partition is also $SUS$ optimal, \cite{framework:1995}, so $k$ is less than or equal to $SUS$. Now, as $n \rightarrow \infty$, $SUS \rightarrow 2\sqrt{n}$. \cite{Aldous99longestincreasing} Thus, $k$ approaches something less than or equal to $2\sqrt{n}$. The worst-case number of comparisons for \textit{delete min} on a standard binary heap is $2\lceil \log{n}-1\rceil$. For the Enc Adaptive List Heap, the worst-case number of comparisons is less than or equal to $2\lceil \log{k}-1\rceil \leq 2\lceil \log{(2\sqrt{n})-1} \rceil = \lceil \log{n} \rceil$. As $n \rightarrow \infty$, the binary heap bound approaches $2\lceil \log{n} \rceil$.
    \end{proof}

\section{Empirical Results} \label{Empirical Results}

    We implemented both the RA List Heap and the EA List Heap for testing purposes, but made no attempts to optimize the code. The results of a series of brief empirical tests against a similarly unoptimized binary heap are presented below. The tests were performed using a codebase written in C and workloads from the 5th DIMACS challenge \cite{DIMACS_5th} (with modifications) and a simple sorting routine. The results presented below are raw wallclock times divided by the minimum time attained by any heap. Thus, 1.00 is the minimum wallclock time and $k$ is $k$ times the minimum.
 
    \begin{table}[h]
        \centering
        \begin{tabular}{ |l|c|c|c| } 
            \hline
             & Sorted Sorting & Random Sorting & Random Dijkstra \\ 
            \hline
            RA List Heap & \textbf{1.00} & 2.29 &\textbf{1.00}\\ 
            EA List Heap &1.15& \textbf{1.00}&1.10 \\
            Binary Heap & 12.62 & 2.52 & 1.22 \\
            \hline
        \end{tabular}
        \caption{Normalized Wallclock Times}
        \label{tab:results}
    \end{table}
    
    Sorted Sorting refers to the task of inserting an $n$-element decreasing sorted sequence into the heap, followed by $n$ consecutive calls to \textit{delete min}. Random sorting refers to the task of inserting an $n$-element random sequence into the heap, followed by $n$ consecutive calls to \textit{delete min}. The results presented for sorting are for $n =$ 10 million. Random Dijkstra refers to Dijkstra workloads from DIMACS. The results presented for Dijkstra are from a strongly connected randomly generated network with 2 million nodes and 8 million edges.
    
    These test results suggest that some variant of the List Heap may be useful in practice. We stress that the results presented above were derived from very simplistic testing. They are merely suggestive that List Heaps may have potential in practice. They are in no way the final word on the empirical performance of the List Heap.
    
\section{Conclusion} \label{Conclusion}
    
    This paper introduced adaptive heaps - heaps whose performance is a function of both the size of the problem instance and the disorder of the problem instance. We introduced the List Heap as a generic structure that can be endowed with adaptive operations. Finally, we presented operations for the List Heap that are adaptive with respect to a number of measures of disorder on the input sequence. 
    
    The discussion within this paper has been relatively informal and the List Heap introduced is far from theoretically optimal. If there is interest in this topic, there are a number of directions from here. Additional empirical testing of the List Heap is needed. We presented two \textit{insert} functions for the List Heap, but there are clearly more (d-array, increasing/decreasing runs, etc.). The optimal choice likely depends on the intended application. If \textit{decrease key} is not needed, List Heaps could be implemented entirely with arrays. On the theory side, the List Heap and discussion of adaptivity included in this paper leaves much to be desired.  We are working on formalizing some notions of adaptivity and have developed a variant of the Fibonacci Heap that is closer to optimal from a theoretical perspective. Modification to other existing heap variants may provide even better results.

\bibliographystyle{plain}
\bibliography{references}

\end{document}